\documentclass[sigconf,nonacm]{acmart}
\renewcommand\footnotetextcopyrightpermission[1]{}
\usepackage{multirow}
\usepackage{algorithm}
\usepackage{algpseudocode}
\usepackage{pifont}
\usepackage{xcolor}
\usepackage{adjustbox}
\usepackage{enumitem}
\usepackage{placeins}

\newcommand{\modelname}{PASRec}

\AtBeginDocument{%
  }
\begin{document}
\title{
When and What to Recommend: Joint Modeling of Timing and Content for Active Sequential Recommendation
}

\author{Jin Chai}
\affiliation{%
  \institution{School of Computing, Macquarie University}
  \city{Sydney}
  \state{NSW}
  \country{Australia}
}
\email{jin.chai@hdr.mq.edu.au}

\author{Xiaoxiao Ma}
\affiliation{%
  \institution{Oracle Analytics Cloud, Oracle Corporation}
  \city{Sydney}
  \state{NSW}
  \country{Australia}
}
\email{xiaoxiao.ma@oracle.com}

\author{Jian Yang}
\affiliation{%
  \institution{School of Computing, Macquarie University}
  \city{Sydney}
  \state{NSW}
  \country{Australia}
}
\email{jian.yang@mq.edu.au}

\author{Jia Wu}
\affiliation{%
  \institution{School of Computing, Macquarie University}
  \city{Sydney}
  \state{NSW}
  \country{Australia}
}
\email{jia.wu@mq.edu.au}

\begin{abstract}
Sequential recommendation aims to model dynamic user preferences from historical interaction sequences and recommend the next item that a user may be interested in.
Most existing works focus on \textit{passive} recommendation, where the system responds only when users actively open the application. Such approaches fail to deliver items of potential interest once the application has been closed. 
Instead, we investigate \textit{active} recommendation, which predicts the user’s next interaction time and actively delivers high-quality recommendations accordingly.
However, active recommendation faces two crucial challenges: (1) Accurately estimating the next interaction time, referred to as the Time of Interest (ToI), as precise timing is essential to avoid delivering recommendations at improper moments or being perceived as spam; and (2) Generating high-quality Item of Interest (IoI) recommendations that are conditioned on the predicted ToI to ensure relevance at the time of delivery.
A typical approach is to first predict the ToI to determine when the user is likely to interact, and then generate the corresponding IoI for recommendation. But, inaccurate ToI prediction induces a single point of failure that propagates errors to IoI estimation and degrades the overall recommendation quality.
We bridge these gaps by proposing Personalized Time of Interest For Active Sequential Recommendation (\modelname), a diffusion-based sequential recommendation framework that aligns ToI and IoI predictions via a tailored training objective, as supported by our theoretical analysis. We further show that \modelname\ is theoretically lower bounded by traditional diffusion-based recommenders.
Through extensive experiments on five widely-used benchmark datasets, we validate the superiority of \modelname\ over eight state-of-the-art baseline models under both leave-one-out and temporal split settings. Further empirical studies on ToI and IoI predictions support our motivation and demonstrate that the performance gain stems from the accurate modeling of both. 

\end{abstract}
\maketitle

\section{Introduction}
\label{sec_Introduction}

\begin{figure}[ht]
    \centering
    \includegraphics[width=\columnwidth]{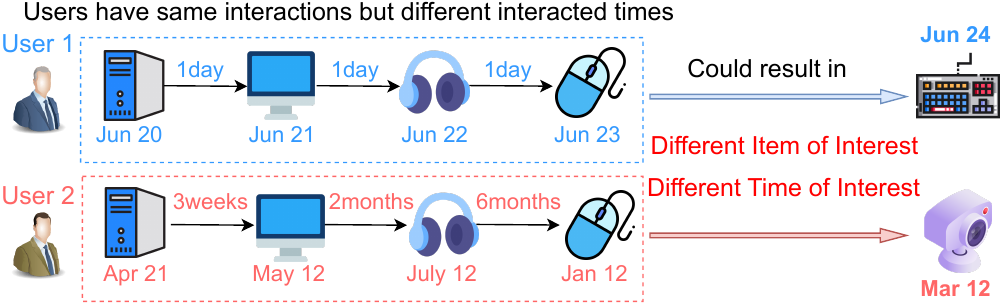}
    \caption{Two users share the same sequence of interactions, but their different interaction times result in distinct next item and time of interest.}
    \label{fig:intro}       
\end{figure}

Sequential recommendation aims to recommend items that best align with users' personalized preferences based on their historical interaction sequences~\cite{pan2024survey, wei2025sequential, li2025dimerec}. 
Existing methods typically operate \textit{passively}, meaning the system only recommends items when users are browsing the application. This approach cannot actively attract users once the application is closed, leading to missed opportunities to deliver potential recommendations and boost revenue. These limitations highlight the need for \textit{active recommendation}, where the system can actively push recommendations~\cite{loni2019personalized}.

For active recommendation scenarios, the accurate prediction of both the Time of Interest (ToI) and the Item of Interest (IoI) is essential, yet presents two major challenges. First, the recommender must precisely estimate the appropriate time to push recommendations in order to avoid sending them at improper moments, which can degrade user experience or be perceived as spam. Second, the recommended items must align with the user’s preferences at the predicted time to ensure relevance and engagement.
This can be formulated as predicting the joint distribution of ToI and IoI, i.e., \( p(e,\tau \mid {s}_u) \). Here, \(\tau\) denotes the ToI, \(e\) is the item to be recommended, and \({s}_u\) represents the user's historical interaction sequence. We can naturally decompose this joint distribution into \( p(e,\tau \mid {s}_u) =p(\tau \mid {s}_u) \cdot p(e \mid \tau, {s}_u)\). The first term \( p(\tau \mid {s}_u) \) represents a precise ToI prediction that enables the model to determine the optimal moment to deliver recommendations. The second term \( p(e \mid \tau, {s}_u) \) denotes the predicted distribution of items.

It is obvious from the second term that user preferences are dependent on time, i.e., time-sensitive~\cite{rahmani2023incorporating}.
In real scenarios, as shown in Figure~\ref{fig:intro}, two users with identical item interaction sequences may exhibit different IoI due to variations in interaction timing. User 1 completed all actions within a week, indicating an urgent intent to assemble a full desktop setup, and is likely to be interested in a keyboard at a near ToI. User 2 completed the same actions but over several months, suggesting occasional upgrades for remote work, making a webcam a more probable next item at a later ToI.

\begin{table*}[ht]
\centering
\caption{Comparison of Temporal and Active Recommendation Capabilities Across Sequential Recommendation Models}
\small
\label{tab:CompareTransposed}
\resizebox{\textwidth}{!}{
\begin{tabular}{lcccccccccc}
\toprule
\textbf{Property} & \textbf{\modelname} & GRU4Rec & SASRec & BERT4Rec & TiSASRec & MEANTIME & TASER & DiffuRec & DreamRec & PreferDiff \\

\midrule
\textbf{Temporal Modeling} & \textcolor{green!60!black}{\ding{51}} & \textcolor{red}{\ding{55}} & \textcolor{red}{\ding{55}} & \textcolor{red}{\ding{55}} & \textcolor{green!60!black}{\ding{51}} & \textcolor{green!60!black}{\ding{51}} & \textcolor{green!60!black}{\ding{51}} & \textcolor{red}{\ding{55}} & \textcolor{red}{\ding{55}} & \textcolor{red}{\ding{55}} \\
\textbf{Encoding Strategy} & Real & \textcolor{red}{\ding{55}} & A & A & Rl + A & Rl + A & Rl + A & A & A & A \\
\textbf{Next ToI Information}       & \textcolor{green!60!black}{\ding{51}} & \textcolor{red}{\ding{55}} &  \textcolor{red}{\ding{55}} & \textcolor{red}{\ding{55}} & \textcolor{red}{\ding{55}} & \textcolor{red}{\ding{55}} & \textcolor{red}{\ding{55}} & \textcolor{red}{\ding{55}} & \textcolor{red}{\ding{55}} & \textcolor{red}{\ding{55}} \\
\textbf{Active Recommendation}       & \textcolor{green!60!black}{\ding{51}} & \textcolor{red}{\ding{55}} &  \textcolor{red}{\ding{55}} & \textcolor{red}{\ding{55}} & \textcolor{red}{\ding{55}} & \textcolor{red}{\ding{55}} & \textcolor{red}{\ding{55}} & \textcolor{red}{\ding{55}} & \textcolor{red}{\ding{55}} & \textcolor{red}{\ding{55}} \\
\bottomrule
\end{tabular}
}
\footnotesize{Note: \textbf{Temporal Modeling} indicates whether the model encodes real interaction timestamps. 
\textbf{Encoding Strategy} describes how temporal or positional information is utilized: 
\textit{A} for encoding absolute positions, 
\textit{Rl} for encoding relative time intervals, 
and \textit{Real} for encoding actual timestamps.  
\textbf{Next ToI Information} indicates whether possible interaction time of next target IoI is considered.
\textbf{Active Recommendation} indicates whether the model can actively deliver recommendations to users.
}
\end{table*}

One may argue that most traditional sequential recommenders can be straightforwardly extended for active recommendation by predicting the ToI first and then predicting IoI. But this naive approach faces a potential single point of failure: First, if the ToI prediction is inaccurate, recommendations could be delivered at inappropriate times, leading users to perceive them as spam. Second, by the nature of \( p(e \mid \tau, {s}_u) \), the accuracy of ToI prediction inherently affects the accuracy of the recommended items, determining how well they capture and reflect users' personalized preferences.

To address these problems, we propose Personalized Time of Interest For Active Sequential Recommendation (\modelname), which adopts a diffusion model as the foundation to model the complex latent distributions of user preferences~\cite{lin2024survey} and incorporates the generated user representations for ToI and IoI prediction. 
With our specially designed training objectives, \modelname: 1) is capable of maximizing the mutual information between historical interaction time and items, which inherently minimizes the risk of single point failure from inaccurate ToI prediction; and 2) generalizable upon our theoretical analysis on the evidence lower bound (ELBO). Specifically, \modelname\ achieves a tighter ELBO compared to traditional diffusion-based sequential recommenders subject to the accuracy of the predicted ToI. 
Our main contributions are:
\begin{itemize}[leftmargin=*,itemsep=0.5pt,topsep=2.5pt]
    \item To the best of our knowledge, \modelname\ is the first to jointly model users' historical time of interest and item of interest for sequential recommendation. This empowers active recommendation to prioritize items according to users' personal preferences in the future, thereby mitigating the risk of spamming and irrelevant promotional pushes.

    \item \modelname\ and its training objectives are specially designed to reinforce the dependency between ToI and IoI, as motivated by our formal analysis. By maximizing their mutual information, \modelname\ effectively learns when and what to recommend and delivers high-quality recommendations actively with regard to the specific time of interest, justified by its tighter ELBO. 
    
    \item Extensive experiments on five widely-used datasets validate \modelname's consistent superiority over eight state-of-the-art baseline methods under both leave-one-out and temporal data splits. Further ablation and case studies highlight the contribution of each key component to the overall performance.
\end{itemize}

\section{Related Works}

In this section, we first review state-of-the-art sequential recommendation (SR) methods, focusing on how user preferences are modeled from historical interaction sequences. We then analyze how temporal information is incorporated, identifying key limitations and motivating the need to jointly model the ToI and IoI.

\subsection{Sequential Recommendation}
Many effective approaches have been proposed to extract and model user preferences to accomplish the SR task~\cite{huang2025survey}. SASRec employs a Transformer architecture to extract a vector embedding that represents user preferences from the historical interaction sequence~\cite{kang2018self}. Similarly, BERT4Rec adopts the same idea but utilizes a bidirectional Transformer encoder~\cite{sun2019bert4rec}. Both methods directly encode user preferences as a single vector and compute the similarity between this vector and the embeddings of candidate items to recommend. However, this straightforward approach can be limited in capturing multiple latent aspects and the diversity of user interests, as representing user preferences with a single embedding vector may oversimplify the complex and multifaceted nature of user behavior~\cite{li2023diffurec}. To better capture the underlying distribution of user preferences, various generative models have been explored, including Generative Adversarial Networks~\cite{jin2020sampling}, Variational Autoencoders~\cite{liang2018variational}, and Diffusion Models~\cite{dhariwal2021diffusion}. Among them, diffusion-based frameworks have shown the most promising results for SR tasks, due to their superior ability to model complex data distributions and enable efficient inference~\cite{wei2025diffusion, lin2024survey}. However, most existing diffusion-based approaches~\cite{li2023diffurec, yang2023generate, mao2025distinguished, liu2024preference} focus solely on modeling user preferences to predict the next IoI, while overlooking the role of temporal signals. In particular, they neglect both past interaction times and the potential next time of the user's interest, which may provide valuable information for capturing personalized user preferences.

\subsection{Temporal Information for SR}
Temporal information is vital in SR tasks, as the timestamps of user-item interactions convey informative cues about user preferences~\cite{dang2023uniform,fan2024tim4rec,li2020time}. Transformer-based models such as SASRec~\cite{kang2018self}, BERT4Rec~\cite{sun2019bert4rec}, DreamRec~\cite{yang2023generate} and PreferDiff~\cite{liu2024preference} adopt learnable absolute positional embeddings to encode the order of items within user's historical interaction sequence. However, relying solely on absolute positional encodings can only capture the sequential ordering information~\cite{li2020time}. This simplification overlooks the fact that variations in interaction timing can reflect distinct user interests. 

To address this limitation, TiSASRec extends the SASRec framework by refining the attention mechanism with a temporal relation matrix, along with both absolute positional and relative time interval embeddings~\cite{li2020time}. Following this idea, subsequent models~\cite{dang2023uniform, zhang2023time, cho2020meantime, tran2023attention, ye2020time, cho2021learning} have been proposed to capture time-aware user behaviors through different architectural designs and embedding strategies. But none of these methods consider predicting the user's next ToI, which is essential for actively delivering recommendations in active sequential recommendation scenarios~\cite{pan2024proactive, bi2024proactive}.

\begin{figure*}[ht]
    \centering
    \includegraphics[width=\textwidth]{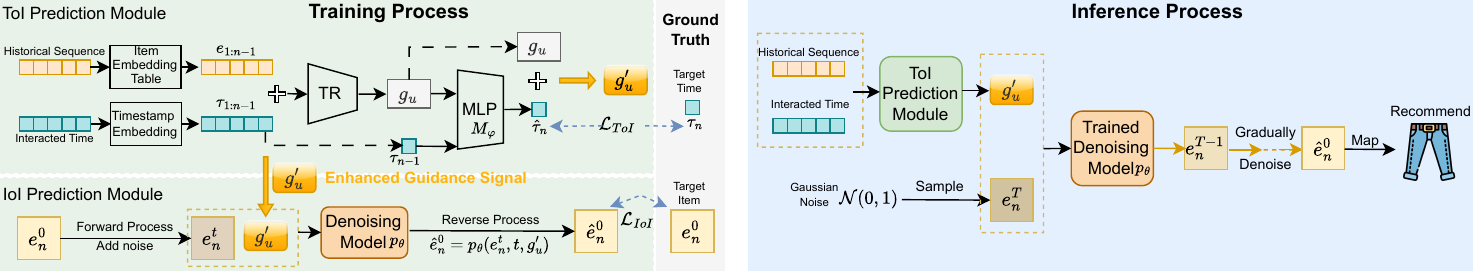}  
    \caption{The framework of \modelname. The left part illustrates how the model jointly trains the ToI and IoI prediction modules. The right part shows how the well-trained model infers the next item of interest based on the user’s historical interaction sequence and the corresponding real interaction times.}  
    \label{fig:framework}
\end{figure*}

As summarized in Table~\ref{tab:CompareTransposed}, most existing methods either rely solely on absolute positional encodings or combine them with relative time intervals, yet none directly utilize actual interaction timestamps. Moreover, to the best of our knowledge, no prior work considers the appropriate timing for delivering recommendations or jointly models the relationship between the next ToI and IoI to support active sequential recommendation.

In this work, we adopt the functions proposed in~\cite{zheng2021rethinking} to encode real timestamps from users' historical interaction sequences into a diffusion-based framework. We introduce a ToI prediction module that estimates the user's next ToI and enriches the user representation with time-dependent information. The IoI prediction module then uses this enriched representation as a guidance signal to guide the diffusion process in generating the IoI. Built upon PreferDiff~\cite{liu2024preference}, our model is theoretically grounded and empirically validated to show that jointly modeling ToI and IoI leads to improved active sequential recommendation performance.

\section{Preliminaries}
\subsection{Task Definition}
Let $\mathcal{U}$ and $\mathcal{V}$ denote the sets of users and items. For each user $u \in \mathcal{U}$, the chronologically ordered sequence of user-item historical interaction sequence is denoted by \(\mathcal{V}_u = [v_1, v_2, \dots, v_{n-1}]\). The sequential recommendation task aims to predict the next item \(\hat{v}_{n}\) that aligns with the user's interests. Through a standard look-up operation with a learnable item embedding table, discrete item IDs are mapped into high-dimensional vectors \( {e} \in \mathbb{R}^{d} \), forming the historical interaction sequence \( {s}_u = [{e}_1, {e}_2, \dots, {e}_{n-1}] \).

\subsection{Diffusion Models}
\subsubsection{Forward Process}
\label{subsubsec:Forward}
The diffusion model gradually adds Gaussian noise to the original target IoI embedding \( e_{n}^{0} \) following a Markov chain \( q(e_{n}^{1:T} \mid e_{n}^{0}) = \prod_{t=1}^{T} q(e_{n}^{t} \mid e_{n}^{t-1}) \), where \( t \in [1, T] \) denotes the diffusion step, and \( [\beta_1, \beta_2, \dots, \beta_T] \) represents the predefined variance schedule~\cite{ho2020denoising}. 

\subsubsection{Reverse Process} 
\label{subsubsec:Reverse}
After the forward process injects Gaussian noise into \( e_{n}^{0} \), the model samples a noisy embedding \( e_{n}^{t} \) at a random timestep \( t \), and a denoising network reconstructs \( e_{n}^{0} \), guided by the user's preference representation \( g_u \) ~\cite{liu2024preference}. Formally, the denoising function from \( e_{n}^t \) to \( e_{n}^{t-1} \) can be defined as:
\begin{equation}
p_\theta(e_{n}^{t-1} \mid e_{n}^t, g_u) = \mathcal{N}(e_{n}^{t-1}; {\mu}_\theta(e_{n}^t, t, g_u), {\Sigma}_\theta(e_{n}^t, t, g_u))
\end{equation}
where \( {\mu}_\theta(e_{n}^t, t, g_u) \) and \( {\Sigma}_\theta(e_{n}^t, t, g_u) \) are the predicted mean and covariance, respectively, parameterized by the diffusion model.

\subsubsection{Inference Process}
\label{subsubsec:inference_process}
After training, the denoising network \( p_\theta \) approximates the reverse diffusion process, 
allowing the recovery of the original data distribution \( p(e_{n}^{0}) \)~\cite{ma2024graph}. For efficient inference, we adopt Denoising Diffusion Implicit Models (DDIM)~\cite{song2020denoising}. Unlike the Markovian reverse process in DDPM, which requires all \(T'\) sequential denoising steps from \(T'\) down to 1, DDIM defines a non-Markovian process that enables accelerated sampling using a reduced set of steps:
\begin{equation}
p_\theta^{\text{DDIM}}(e_n^{0:T}) 
= p_\theta(e_n^T) 
\prod_{i=1}^{T} 
p_\theta^{\text{DDIM}}\big(e_n^{t_{i-1}} \mid e_n^{t_i}, g_u \big),
\end{equation}
where \(T \ll T'\) and \(\{t_i\}_{i=0}^{T}\) denotes the timesteps for sampling.

\section{Methodology}

\subsection{Overview of \modelname}
We introduce \modelname, which consists of two key components: a ToI prediction module and an IoI prediction module. The overall design is inspired by extracting a personalized user representation as a guidance signal and adopting the diffusion-based item generation process proposed in~\cite{yang2023generate, liu2024preference}. In the following sections, we detail the components and workflow of \modelname.

\subsection{Time Encoding Functions}
\label{subsec:TEM}
In order to extract temporal information from historical interaction sequences, we have utilized several functions to encode real timestamps. \cite{zheng2021rethinking} challenges the exclusive use of sinusoidal functions for positional encoding and suggests that other continuous basis functions, such as Gaussian kernel functions and Random Fourier Features, can also encode positional or temporal information. We treat different encoding functions as hyperparameters, allowing the model to encode real interaction timestamps, as shown in Figure~\ref{fig:EncodingFunctions}.

\begin{figure}[ht]
    \centering
    \includegraphics[width=\columnwidth]{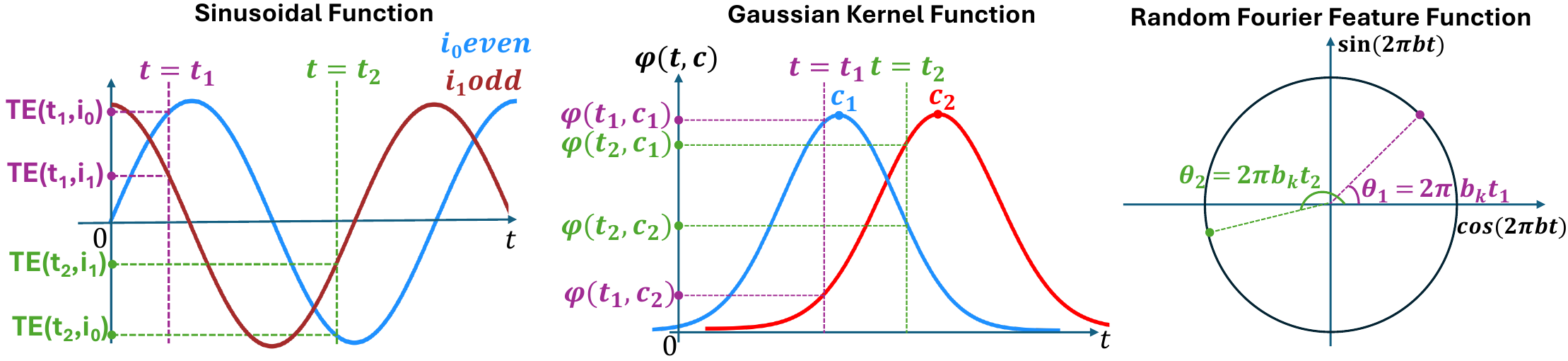}  
    \caption{Overview of Time Encoding Functions. }
    
    \label{fig:EncodingFunctions}
\end{figure}

\subsubsection{Sinusoidal Function} This idea was originally introduced in the Transformer model~\cite{vaswani2017attention} to incorporate positional information into sequence modeling. In our work, we apply a similar concept by encoding the interaction timestamps into high-dimensional embedding representations, the formula is defined as:
\begin{equation}
\begin{aligned}
\mathrm{TE}(t,2i)
&= \sin\Bigl(\tfrac{t}{freq^{\,2i/d}}\Bigr),
\mathrm{TE}(t,2i+1) = \cos\Bigl(\tfrac{t}{freq^{\,2i/d}}\Bigr) .
\end{aligned}
\end{equation}
where \( freq \) is a frequency typically set to 10000, and \( d \) is the total embedding dimension, with \( i = 0, \dots, \frac{d}{2} - 1 \).

\subsubsection{Gaussian Kernel Function}
The Gaussian embedder is defined as \( \psi(t, c) = \exp\left(-\frac{\|t - c\|^2}{2\sigma^2}\right) \), where \( \sigma \) is a tunable hyperparameter representing the standard deviation~\cite{zheng2021rethinking}. We define: 
\begin{equation}
\mathrm{TE}(t) = [\psi_0(t, c_0), \psi_1(t, c_1), \dots, \psi_{d-1}(t, c_{d-1})] 
\end{equation}
where each \( c_j \) is the center of the \( j \)-th Gaussian kernel, given by \( c_j = \frac{j}{d - 1} \) for \( j = 0, 1, \dots, d - 1 \). The Gaussian embedder constructs a d-dimensional embedding for each timestamp by applying a Gaussian kernel function to compute its similarity with a set of fixed centers.

\subsubsection{Random Fourier Feature (RFF)}
This method can employ fourier frequency to map low-dimensional timestamps into a higher-dimensional representation~\cite{zheng2021rethinking}. Specifically, \( t \) is encoded by:
\begin{equation}
\mathrm{TE}(t) = \left[ \cos\left(2\pi b_k t\right), \sin\left(2\pi b_k t\right) \right]
\end{equation}
where the frequencies \( b_k \) are independently sampled for every dimension from a Gaussian distribution \(b_k \sim \mathcal{N}(0, \sigma^2), \quad k = 1, 2, \dots, d/2\), \( \sigma \) is the tunable standard deviation. 

Through the time encoding function, the real interacted timestamps \({{t}}_u\) can be encoded into high-dimensional vector embeddings:
\begin{equation}
\text{TE}({{t}}_u) = [{\tau}_1, {\tau}_2, \dots, {\tau}_{n-1}]
\end{equation}
We then incorporate the temporal information into the user's historical interaction sequence by combining each item embedding with its corresponding time embedding, resulting in:
\begin{equation}
\tilde{{s}}_u = [e_1 + {\tau}_1, e_2 + {\tau}_2, \dots, e_{n-1} + {\tau}_{n-1}].
\end{equation}

\subsection{ToI Prediction Module}
\label{subsec:TPM}
To enrich the user representation with the predicted time of the next potential interaction, we introduce the ToI Prediction Module. After we obtain the historical sequence item embeddings \(\tilde{{s}}_u\), which have been enriched with past temporal information. A Transformer encoder (\textit{TR}) is then employed to extract the user representation, defined as \( {g}_u = \textit{TR}(\tilde{{s}}_u) \).

The ToI prediction module \( M_\varphi \) utilizes the user representation \( {g}_u \) along with the embedding of the timestamp corresponding to the most recent interaction \( \tau_{n-1} \) to predict the ToI embedding of the forthcoming IoI. Formally,
\begin{equation}
\hat\tau_{n} = M_\varphi({g}_u, \tau_{n-1})
\end{equation}
The predicted ToI embedding and the original user representation are fused through a multilayer perceptron (MLP) to produce an enriched user representation \( {g}_u' \). Formally,
\begin{equation}
\label{equation_fusion}
{g}_u' = \mathrm{MLP}({g}_u, \hat\tau_{n})
\end{equation}
This fused representation is subsequently combined with the original representation using a residual connection, updated as:
\begin{equation}
{g}_u' \leftarrow (1 - \gamma)\,{g}_u + \gamma\,{g}_u'
\end{equation}
Here, \( \gamma \) is a tunable weight that balances the contribution of the predicted ToI information to the user representation. A larger value of \( \gamma \) indicates that more predicted temporal information is integrated, whereas a smaller value implies less influence.

\subsection{IoI Prediction Module}
\label{subsec:IPM}
After extracting the enriched guidance signal \( {g}_u' \) with the ToI information as discussed in Section~\ref{subsec:TEM} and ~\ref{subsec:TPM}, \modelname\ follows the idea of adopting a conditional diffusion-based framework~\cite{yang2023generate, mao2025distinguished, liu2024preference} to train the denoising model \( p_\theta \). The denoising model, implemented as a simple MLP, is trained to reconstruct the clean target IoI embedding conditioned on the enriched guidance signal \( {g}_u' \).

First, we add Gaussian noise to the original target IoI embedding during the forward process, as we discussed in~\ref{subsubsec:Forward}:
\begin{equation}
e_{n}^{t} = \sqrt{\bar{\alpha}_t} \, e_{n}^{0} + \sqrt{1 - \bar{\alpha}_t} \, {\epsilon}, \quad t \in [1, T], \quad{\epsilon} \sim \mathcal{N}(0, {I})
\end{equation}
At each diffusion step \( t \), the objective of the denoising model is to estimate the clean target IoI embedding \( \hat{e}_{n}^{0} \). Following prior works~\cite{li2023diffurec, yang2023generate, mao2025distinguished, liu2024preference}, we formulate this process as:
\begin{equation}
\label{equation:denoise}
\hat{e}_{n}^{0} = p_\theta(e_{n}^{t}, t, {g}_u')
\end{equation}

\subsection{Training Criteria}
\label{subsection_training}
The overall objective of \modelname\ \( \mathcal{L}_{\text{\modelname}} \) is to learn the joint distribution \( p(e, \tau \mid {s}_u) \) by optimizing two components: the ToI prediction loss \( \mathcal{L}_{\text{ToI}} \), which encourages accurate estimation of the next ToI embedding, and the IoI prediction loss \( \mathcal{L}_{\text{IoI}} \), which guides the model to reconstruct the IoI embedding based on the enriched guidance signal incorporating the predicted ToI information.

\paragraph{ToI Prediction Loss}
The loss for the ToI prediction module is defined using the cosine similarity \( S(\cdot, \cdot) \) between the predicted ToI embedding and the ground-truth ToI embedding:
\begin{equation}
\label{equation_Ltime}
\mathcal{L}_{\text{ToI}} = -S(\tau_{n},M_\varphi({g}_u, \tau_{n-1}))
\end{equation}

\paragraph{IoI Prediction Loss}
We adopt a BPR-style positive–negative contrastive learning strategy to train the denoising model, following~\cite{liu2024preference}, in order to better capture user preferences. The loss for the standard target IoI denoising task is defined as:
\begin{equation}
\label{equation_Lnormal}
\mathcal{L}_{\text{Normal}} = \mathbb{E}_{({e}_{n}^{0+}, t, {g}_u')} \left[ \left\| p_{\theta}({e}_{n}^{t+}, t, {g}_u') - {e}_{n}^{0+} \right\|_2^2 \right]
\end{equation}
, where \( {e}_{n}^{0+} \) denotes the positive item that the user has actually interacted with. Then we randomly sample a set of \( k \) negative items from the current training batch, denoted as \( \mathcal{H} = \{{e}_{n}^{0-(1)},  \dots, {e}_{n}^{0-(k)}\} \). We compute the centroid of the negative set as \( {{e'}}_{n}^{0-} = \frac{1}{k} \sum_{i=1}^{k} {e}_{n}^{0-(i)} \), aiming to mitigate the drawbacks of individual negative sampling~\cite{liu2024preference}. We formulate the BPR loss to enhance the likelihood of positive items while pushing them away from the centroid of the sampled negative items, where \( \sigma(\cdot) \) represents the sigmoid function:
\begin{equation}
\label{equation_LBPR}
\mathcal{L}_{\text{BPR}} = -\log \sigma \left( -k \cdot \left[ S(\hat{{e}}_{n}^{0+}, {e}_{n}^{0+}) - S(\hat{{e'}}_{n}^{0-}, {{e'}}_{n}^{0-}) \right] \right)
\end{equation} 
Then, we formulate the IoI prediction loss for denoising the target IoI embedding in the diffusion process as follows:
\begin{equation}
\label{equation_LDiff}
\mathcal{L}_{\text{IoI}} = \lambda \mathcal{L}_{\text{Normal}} + (1 - \lambda)\mathcal{L}_{\text{BPR}}
\end{equation}
where \( \lambda \in [0,1] \) is a tunable hyperparameter that balances the objectives of generation learning and preference modeling.

Above all, we define the final loss function for \modelname\ as:
\begin{equation}
\label{equation_LTAGDiff}
\mathcal{L}_{\text{\modelname}} = \eta \mathcal{L}_{\text{IoI}} + (1 - \eta)\mathcal{L}_{\text{ToI}}
\end{equation}
where \( \eta \in [0,1] \) is a tunable hyperparameter that balances the model's ability to denoise and reconstruct the target IoI embedding and to predict the potential ToI embedding.

\subsection{Inference Criteria}
To accelerate the inference process, we adopt DDIM~\cite{song2020denoising} for the inference process instead of DDPM, following~\cite{liu2024preference}. The forward process follows the formulation from Section~\ref{subsubsec:Forward}:
\begin{equation}
{e}_n^{t} = \sqrt{\bar{\alpha}_t} \, {e}_n^{0} + \sqrt{1 - \bar{\alpha}_t} \, {\epsilon}, \quad t \in [1, T]
\end{equation}
where the noise \( {\epsilon} \) can be derived as:
\begin{equation}
{\epsilon} = \frac{{e}_n^{t} - \sqrt{\bar{\alpha}_t} \, {e}_n^{0}}{\sqrt{1 - \bar{\alpha}_t}}
\end{equation}

As mentioned in Section~\ref{subsec:IPM}, the model predicts \( \hat e_n^0 \) at each diffusion step \( t \). 
We employ classifier-free guidance~\cite{ho2022classifier} as:
\begin{equation}
\hat e_n^0 = (1 + w) \, p_\theta(e_n^t, t, g_u') - w \, p_\theta(e_n^t, t, \Phi)
\end{equation}
where \( w \) is the guidance weight and \( \Phi \) denotes a dummy token representing the unconditional path.

Accordingly, the noise and the next-step \( {e}_n^{t-1} \) can be directly estimated using the predicted clean IoI embedding \( \hat{{e}}_n^{0} \) and the approximated noise \( \hat{{\epsilon}} \) as follows, without introducing additional randomness, thus making the sampling process more efficient:
\begin{equation}
\label{equation_denoise}
\hat{{\epsilon}} = \frac{{e}_n^{t} - \sqrt{\bar{\alpha}_t} \, \hat{{e}}_n^{0}}{\sqrt{1 - \bar{\alpha}_t}}, \quad
{e}_n^{t-1} = \sqrt{\bar{\alpha}_{t-1}} \, \hat{{e}}_n^{0} + \sqrt{1 - \bar{\alpha}_{t-1}} \, \hat{{\epsilon}}
\end{equation}

Based on this idea, we start by sampling a pure Gaussian noise as the initial target IoI embedding at timestep \( T \), \( {e}_n^{T} \sim \mathcal{N}(0, {I}) \). At each timestep, the denoising model predicts \( \hat{{e}}_n^{0} \), which is then used to compute the previous-step representation \( {e}_n^{t-1} \) via Equation~\ref{equation_denoise}. This process iterates from \( t = T \) down to \( t = 0 \), yielding the final reconstructed item embedding \( {e}_n^{0} \). Finally, we compute the similarity between the final predicted representation \( {e}_n^{0} \) and all candidate item embeddings. The top-K most similar items are then selected as the final recommendation results for the user. We summarize our algorithms in Appendix~\ref{app:alg}.

\subsection{Theoretical analysis}

\begin{proposition}
Minimizing the loss function of \modelname\ implicitly encourages the model to maximize the mutual information between the ToI and IoI representations.
\end{proposition}

\begin{proof}
From Equation~\ref{equation_LTAGDiff}, the overall loss of \modelname\ is composed of two components: IoI prediction loss $\mathcal{L}_{\text{IoI}}$ and ToI prediction loss $\mathcal{L}_{\text{ToI}}$. $\mathcal{L}_{\text{ToI}}$ aims to predict the ToI embedding $\hat{\tau}_n$ conditioned on ${g}_u$. Minimizing this objective implicitly increases the mutual information between $\tau_n$ and ${g}_u$. $\hat{\tau}_n$ is then fused with ${g}_u$ via Equation~\ref{equation_fusion} to produce the enhanced user representation ${g}_u'$, which is used to predict the IoI representation ${e}_n$. The loss $\mathcal{L}_{\text{IoI}}$ is formulated as a BPR-style contrastive objective, and contrastive learning objective is known to maximize a lower bound of mutual information~\cite{oord2018representation}. So this implicitly increases the mutual information between ${g}_u'$ and ${e}_n$. As both ToI ($\tau_n$) and IoI (${e}_n$) interact with the shared user representation, minimizing the total loss encourages the model to capture more mutual information between them.
\end{proof}

\begin{proposition}
The Evidence Lower Bound (ELBO) of \modelname\ is greater than or equal to that of traditional diffusion-based recommenders, subject to the accurate prediction of ToI.
\end{proposition}

\begin{proof}
We first recall the ELBO of diffusion models for sequential recommendation, as derived in ~\cite{luo2022understanding, yang2023generate}:
\begin{align}
\log p({e}_n) 
&= \mathbb{E}_{q({e}_n^{0:T})} \left[ \log p({e}_n^{0}) \right] 
+ D_{\mathrm{KL}} \left( q({e}_n^{0:T}) \,\|\, p({e}_n^{0:T}) \right) \\
&\geq \mathbb{E}_{q({e}_n^{0:T})} \left[ \log p({e}_n^{0}) \right] 
\end{align}
Theoretically, a larger ELBO suggests that the model is more likely to achieve a closer optimization of the true data likelihood~\cite{ho2020denoising}. Traditional diffusion-based sequential recommenders do not incorporate ToI information. Their ELBO is defined as
\begin{equation}
\mathrm{ELBO}_o =\mathbb{E}_{q({e}_n^{0:T})} \left[ \log p({e}_n^0 \mid {g}_u) \right],
\label{eq:elbo_without_time}
\end{equation}
which depends solely on the user representation \( {g}_u \) extracted from the historical interaction sequence \( {s}_u \) (see Section~\ref{subsec:TPM}).

\modelname\ jointly models the target IoI and ToI information, resulting in the following ELBO over the joint distribution:
\begin{equation}
\mathrm{ELBO}_t = \mathbb{E}_{q(e_n^{0:T})} \left[ \log p({e}_n^0, {\tau}_n \mid {g}_u) \right]
\label{eq:elbo_joint}
\end{equation}

Based on the chain rule of joint probability, the ELBO over the joint distribution \( p({e}_n^0, {\tau}_n \mid {g}_u) \) can be decomposed as:
\begin{equation}
\begin{aligned}
\label{equation_ELBOTAGDiff}
\mathrm{ELBO}_t 
&= \mathbb{E}_{q} \left[ \log p({e}_n^0, {\tau}_n \mid {g}_u) \right] \\
&= \mathbb{E}_{q} \left[ \log p_\theta({e}_n^0 \mid {\tau}_n, {g}_u) + \log M_\varphi({\tau}_n \mid {g}_u) \right] \\
&= \underbrace{\mathbb{E}_{q} \left[ \log p_\theta({e}_n^0 \mid {\tau}_n, {g}_u) \right]}_{\text{IoI Prediction}}
\; + \;
\underbrace{\mathbb{E}_{q} \left[ \log M_\varphi({\tau}_n \mid {g}_u) \right]}_{\text{ToI Prediction}}
\end{aligned}
\end{equation}

Based on the definitions of conditional entropy and mutual information~\cite{cover1999elements}, we have:
\begin{equation}
I({e}_n^0, {\tau}_n \mid {g}_u)
= H({e}_n^0 \mid {g}_u) - H({e}_n^0 \mid {\tau}_n, {g}_u)
\label{eq:mutual_info_def}
\end{equation}
Since mutual information is always non-negative:
\begin{equation}
I({e}_n^0, {\tau}_n \mid {g}_u) \geq 0 
\quad \Rightarrow \quad
H({e}_n^0 \mid {\tau}_n, {g}_u) \leq H({e}_n^0 \mid {g}_u)
\label{eq:entropy_inequality}
\end{equation}
By the definition of conditional entropy:
\begin{equation}
H({e}_n^0 \mid {\tau}_n, {g}_u)
= - \mathbb{E}_q \left[ \log p({e}_n^0 \mid {\tau}_n, {g}_u) \right]
\label{eq:cond_entropy_tau}
\end{equation}
\begin{equation}
H({e}_n^0 \mid {g}_u)
= - \mathbb{E}_q \left[ \log p({e}_n^0 \mid {g}_u) \right]
\label{eq:cond_entropy_no_tau}
\end{equation}
Therefore, we conclude:
\begin{equation}
\underbrace{\mathbb{E}_{q} \left[ \log p_\theta({e}_n^0 \mid {\tau}_n, {g}_u) \right]}_{\text{IoI Prediction}}
\geq
\underbrace{\mathbb{E}_q \left[ \log p({e}_n^0 \mid {g}_u) \right]}_{\mathrm{ELBO}_o},
\label{eq:loglikelihood_comparison}
\end{equation}

And for the second term of Equation~\ref{equation_ELBOTAGDiff}, this approaches zero when the ToI prediction model $M_\varphi$ approximates the true posterior distribution well. Therefore, under the condition that ToI prediction is accurate, we conclude that the following inequality holds:
\begin{equation}
\underbrace{\mathbb{E}_{q} \left[ \log p({e}_n^0, {\tau}_n \mid {g}_u) \right]}_{\mathrm{ELBO}_t}
\geq
\underbrace{\mathbb{E}_q \left[ \log p({e}_n^0 \mid {g}_u) \right]}_{\mathrm{ELBO}_o},
\label{eq:loglikelihood_comparison}
\end{equation}

Section~\ref{ToIaccuracy} presents empirical results demonstrating that the ToI prediction model achieves high accuracy in practice.
\end{proof}

\section{Experiments}
We design experiments to answer the following research questions:

\textbf{RQ1:} 
How effective is \modelname\ compared to the state-of-the-art methodologies?

\textbf{RQ2:} 
How do different time encoding functions impact the performance of \modelname? 

\textbf{RQ3:} 
How do the designed components contribute to the performance improvements of \modelname?

\textbf{RQ4:} 
How does incorporating predicted ToI information into user preferences influence the performance of \modelname?

\textbf{RQ5:} 
How accurately does the ToI Prediction Module estimate the user's time of interest compared to the ground truth?

\begin{table}[t]
\caption{Statistics of Datasets}
\label{tab:Statical}
\centering
\resizebox{\columnwidth}{!}{
\begin{tabular}{lrrrrr}
\toprule
\textbf{Dataset} & \textbf{\# Sequence} & \textbf{\# items} & \textbf{\# Actions} & \textbf{Avg\_len} & \textbf{Sparsity} \\
\midrule
Beauty & 22,363 & 12,101 & 162,150 & 7.25 & 99.94\% \\
Toys & 19,412 & 11,921 & 138,444 & 7.13 & 99.94\% \\
Sports & 35,598 & 18,357 & 256,598 & 7.21 & 99.96\% \\
MovieLens-10M & 69,878 & 10,027 & 3,054,340 & 43.71 & 99.56\% \\
MovieLens-20M & 138,493 & 17,177 & 6,013,602 & 43.42 & 99.75\% \\

\bottomrule
\end{tabular}
}
\end{table}

\begin{table*}[htbp]
\small
\centering
\caption{Performance comparison of models under LOO Split (\%). Bold denotes the best; underline denotes the second best.}
\label{tab:LOO_compare}
\begin{adjustbox}{max width=\textwidth}
\begin{tabular}{llccccccccccc}
\toprule
\textbf{Dataset} & \textbf{Metric} & \textbf{SASRec} & \textbf{GRU4Rec} & \textbf{BERT4Rec} & \textbf{TiSASRec} & \textbf{MEANTIME} & \textbf{DiffuRec} & \textbf{DreamRec} & \textbf{PreferDiff} & \textbf{\modelname} \\
\midrule

\multirow{4}{*}{Beauty}
& \textbf{H@5}     & 0.80 & 1.73 & 1.21 & 2.24  & 1.94 & 1.61 & 2.28 & \underline{3.00} & \textbf{3.43} \\
& \textbf{H@10}    & 1.47 & 2.81 & 2.22 & 3.52  & 3.15 & 2.37 & 3.25 & \underline{3.81} & \textbf{4.31} \\
& \textbf{N@5}     & 0.43 & 1.14 & 0.77 & 1.39  & 1.21 & 0.99 & 1.55 & \underline{2.23} & \textbf{2.48} \\
& \textbf{N@10}    & 0.65 & 1.49 & 1.09 & 1.81  & 1.61 & 1.23 & 1.86 & \underline{2.49} & \textbf{2.77} \\
\midrule

\multirow{4}{*}{Sports}
& \textbf{H@5}     & 0.63 & 0.58 & 0.82 & 1.14 & 0.97 & 0.65 & 1.05 & \underline{1.50} & \textbf{1.66} \\
& \textbf{H@10}    & 1.29 & 1.14 & 1.56 & 1.73 & 1.51 & 1.15 & 1.49 & \underline{1.93} & \textbf{2.14} \\
& \textbf{N@5}     & 0.41 & 0.41 & 0.50 & 0.74 & 0.60 & 0.44 & 0.69 & \underline{1.06} & \textbf{1.21} \\
& \textbf{N@10}    & 0.62 & 0.59 & 0.73 & 0.92 & 0.77 & 0.60 & 0.84 & \underline{1.20} & \textbf{1.36} \\
\midrule

\multirow{4}{*}{Toys}
& \textbf{H@5}     & 0.64 & 0.81 & 0.71 & 1.69 & 0.69 & 0.09 & 2.66 & \underline{3.21} & \textbf{3.39} \\
& \textbf{H@10}    & 1.27 & 1.31 & 1.10 & 2.65 & 1.34 & 0.21 & 3.46 & \underline{3.66} & \textbf{3.88} \\
& \textbf{N@5}     & 0.37 & 0.52 & 0.43 & 1.03 & 0.45 & 0.05 & 1.88 & \underline{2.49} & \textbf{2.65} \\
& \textbf{N@10}    & 0.57 & 0.68 & 0.56 & 1.34 & 0.66 & 0.09 & 2.14 & \underline{2.63} & \textbf{2.81} \\
\midrule

\multirow{4}{*}{ML-10M}

& \textbf{H@5}     & 2.86  & 6.62  & 1.17  & 4.45  & 5.23  & \underline{6.65}  & 3.82  & 6.48 & \textbf{8.89} \\
& \textbf{H@10}    & 5.30  & 10.66 & 2.32  & 7.94  & 9.15  & \underline{10.83} & 5.48  & 8.98 & \textbf{11.67} \\
& \textbf{N@5}     & 1.68  & 4.25  & 0.67  & 2.69  & 3.27  & 4.21  & 2.75  & \underline{4.48} & \textbf{6.43} \\
& \textbf{N@10}    & 2.46  & \underline{5.55}  & 1.04  & 3.81  & 4.52  & \underline{5.55}  & 3.28  & 5.28 & \textbf{7.33} \\
\midrule

\multirow{4}{*}{ML-20M}

& \textbf{H@5}     & 2.28  & 7.38  & 1.26  & 4.51  & 5.67  & \underline{7.49}  & 3.58 & 6.32  & \textbf{8.96} \\
& \textbf{H@10}    & 4.52  & 11.50 & 2.46  & 7.96  & 9.55  & \underline{11.93} & 4.99 & 8.83  & \textbf{11.95} \\
& \textbf{N@5}     & 1.33  & 4.76  & 0.75  & 2.71  & 3.56  & \underline{4.82}  & 2.55 & 4.33  & \textbf{6.39} \\
& \textbf{N@10}    & 2.05  & 6.08  & 1.13  & 3.81  & 4.80  & \underline{6.25}  & 3.01 & 5.14  & \textbf{7.35} \\

\bottomrule
\end{tabular}
\end{adjustbox}
\end{table*}

\begin{table*}[htbp]
\centering
\small
\caption{Performance comparison of models under Temporal Split (\%). Bold denotes the best; underline denotes the second best.}
\label{tab:811_compare}
\begin{adjustbox}{max width=\textwidth}
\begin{tabular}{llccccccccccc}
\toprule
\textbf{Dataset} & \textbf{Metric} & \textbf{SASRec} & \textbf{GRU4Rec} & \textbf{BERT4Rec} & \textbf{TiSASRec} & \textbf{MEANTIME} & \textbf{DiffuRec} & \textbf{DreamRec} & \textbf{PreferDiff} & \textbf{\modelname}  \\
\midrule

\multirow{4}{*}{Beauty}

& \textbf{H@5}     & 0.89 & 1.88 & 2.15  & 2.06  & 2.32  & 2.32 & 3.80 & \underline{4.07} & \textbf{4.74} \\
& \textbf{H@10}    & 1.70 & 2.59 & 3.53  & 4.16  & 3.67  & 3.49 & 4.69 & \underline{4.74} & \textbf{5.63} \\
& \textbf{N@5}     & 0.55 & 1.15 & 1.31  & 1.41  & 1.50  & 1.50 & 2.84 & \underline{3.10} & \textbf{3.58} \\
& \textbf{N@10}    & 0.82 & 1.38 & 1.75  & 2.08  & 1.92  & 1.89 & 3.13 & \underline{3.32} & \textbf{3.87} \\

\midrule

\multirow{4}{*}{Sports}

& \textbf{H@5}     & 0.84 & 1.12 & 1.35 & 1.29 & 0.95 & 1.01 & 1.54 & \underline{1.85} & \textbf{2.11} \\
& \textbf{H@10}    & 1.63 & 1.71 & \textbf{2.50} & 1.71 & 1.74 & 1.60 & 1.97 & \underline{2.19} & \textbf{2.50} \\
& \textbf{N@5}     & 0.50 & 0.67 & 0.93 & 0.79 & 0.55 & 0.56 & 1.23 & \underline{1.47} & \textbf{1.65} \\
& \textbf{N@10}    & 0.75 & 0.87 & 1.30 & 0.92 & 0.80 & 0.74 & 1.36 & \underline{1.58} & \textbf{1.78} \\

\midrule

\multirow{4}{*}{Toys}

& \textbf{H@5}     & 0.72 & 1.60 & 0.77 & 1.39 & 1.70 & 2.01 & \underline{4.43} & 4.27 & \textbf{4.94} \\
& \textbf{H@10}    & 0.98 & 2.32 & 1.60 & 2.52 & 2.06 & 2.73 & \underline{5.51} & 5.05 & \textbf{5.61} \\
& \textbf{N@5}     & 0.58 & 1.14 & 0.55 & 1.08 & 1.12 & 1.47 & \underline{3.37} & 3.31 & \textbf{3.84} \\
& \textbf{N@10}    & 0.66 & 1.37 & 0.83 & 1.44 & 1.24 & 1.70 & \underline{3.73} & 3.56 & \textbf{4.06} \\

\midrule

\multirow{4}{*}{ML-10M}
& \textbf{H@5}     & 3.78  & 7.24  & 5.05  & 6.08  & 7.20  & \underline{9.89}  & 5.17  & 8.67  & \textbf{11.00} \\
& \textbf{H@10}    & 6.47  & 10.96 & 9.40  & 10.02 & 10.66 & \underline{13.75} & 7.11  & 11.55 & \textbf{13.79} \\
& \textbf{N@5}     & 2.66  & 4.89  & 3.15  & 4.04  & 4.93  & \underline{7.18}  & 3.82  & 6.32  & \textbf{8.31} \\
& \textbf{N@10}    & 3.51  & 6.09  & 4.54  & 5.30  & 6.05  & \underline{8.43}  & 4.45  & 7.25  & \textbf{9.21} \\

\midrule

\multirow{4}{*}{ML-20M}

& \textbf{H@5}     & 2.35  & 8.11  & 3.96  & 6.38  & 7.89  & \underline{10.13}  & 4.36  & 8.19  & \textbf{10.56} \\
& \textbf{H@10}    & 4.28  & 11.92 & 7.13  & 9.81  & 11.60 & \textbf{14.12}  & 6.04  & 10.92 & \underline{13.41} \\
& \textbf{N@5}     & 1.49  & 5.60  & 2.30  & 4.13  & 5.37  & \underline{7.27}   & 3.16  & 5.98  & \textbf{7.81}  \\
& \textbf{N@10}    & 2.11  & 6.83  & 3.31  & 5.24  & 6.56  & \underline{8.55}   & 3.70  & 6.87  & \textbf{8.74}  \\
\bottomrule
\end{tabular}
\end{adjustbox}
\end{table*}

\subsection{Experimental Settings}
\subsubsection{Dataset Selection}
We evaluate the effectiveness of \modelname\ using five commonly used and publicly available datasets:

\textbf{Amazon 2014\footnotemark}: 
We select the \textit{Sports}, \textit{Beauty}, and \textit{Toys} categories within the Amazon review dataset, collected from the Amazon online shopping platform from May 1996 to July 2014. 
\footnotetext{\url{https://cseweb.ucsd.edu/~jmcauley/datasets/amazon/links.html}}

\textbf{MovieLens\footnotemark}:  
We select the \textit{MovieLens 10M, 20M} subsets, which contain approximately 10 and 20 million user ratings for movies.
\footnotetext{\url{https://grouplens.org/datasets/movielens/}}

Detailed statistics of these datasets are provided in Table~\ref{tab:Statical}.

\subsubsection{Baseline Models}

To demonstrate the efficacy of \modelname\, we compared it with eight state-of-the-art methodologies, which can be categorized into:

\begin{itemize}[leftmargin=*, labelsep=0.5em]
    \item \textbf{Traditional recommenders:} 
    GRU4Rec~\cite{hidasi2015session}, SASRec~\cite{kang2018self}, and BERT4Rec~\cite{sun2019bert4rec} employ discriminative architectures such as GRU and Transformer to model sequential dependencies.

    \item \textbf{Time-aware recommenders:} 
    TiSASRec~\cite{li2020time} introduces a time interval aware self-attention mechanism, and MEANTIME~\cite{cho2020meantime} uses mixture of temporal embeddings into attention mechanism.

    \item \textbf{Generative recommenders:} 
    DreamRec~\cite{hu2024generate}, DiffuRec~\cite{li2023diffurec}, and PreferDiff~\cite{liu2024preference} adopt conditional diffusion-based architectures, where user representations are extracted from historical interactions and used as guidance signals to recommend.
\end{itemize}

\subsubsection{Implementation Details}
Following prior works~\cite{kang2018self, li2023diffurec, liu2024preference}, we treat all user–item interactions as implicit positive feedback and others as negative. All users' historical interaction sequences are chronologically ordered, and users or items with fewer than five interactions are filtered out. We limit sequence lengths to 50 for MovieLens and 10 for Amazon, padding shorter sequences with a special token. Raw timestamps are converted to day-level granularity, then shifted to start from zero and normalized to $[0,1]$ to enable the model to effectively learn temporal patterns. To ensure comprehensive evaluation, we adopt two widely-used data partitioning strategies: the 8:1:1 temporal split~\cite{yang2023generate,liu2024preference} and the leave-one-out (LOO) strategy~\cite{kang2018self,sun2019bert4rec,li2023diffurec}. We evaluate model performance using Hit Rate (H@K) and NDCG (N@K) with \( K = \{5, 10\} \).
Due to space limitation, implementation details of all baseline methods and our hyper-parameter settings are provided in Appendix~\ref{app:exp} and \ref{app:hp}.

\subsection{Overall Performance (RQ1)}
\subsubsection{Recommendation Performance Analysis} 
We conduct experiments on five benchmark datasets to evaluate the recommendation performance of \modelname\ against eight baselines. The results are presented in Table~\ref{tab:LOO_compare} and Table~\ref{tab:811_compare}. Overall, diffusion-based models such as DiffuRec, DreamRec, and PreferDiff generally outperform traditional methods, confirming their strong capability in modeling user preference distributions. Remarkably, \modelname\ achieves the best performance across almost all benchmark datasets and both LOO and temporal data partitions. These results empirically validate our theoretical motivation for jointly modeling users’ Time of Interest and Item of Interest. By leveraging personalized ToI information to enrich user preference representations, \modelname\ provides more accurate IoI predictions for active sequential recommendation.

\subsubsection{Computational Complexity Analysis}
We analyze the computational efficiency of \modelname\ under the LOO partition, and the results are shown in Table~\ref{tab:comparison_Complexity}. Although \modelname\ incorporates additional ToI prediction, the training time increases only marginally. More importantly, it achieves improved recommendation performance without significantly increasing inference time. This demonstrates \modelname's ability to deliver high-quality recommendations efficiently, making it suitable for sequential recommendation tasks where low-latency inference is essential.

\begin{table}[htbp]
\centering
\caption{Time Cost per Training Epoch/Inference Sample}
\label{tab:comparison_Complexity}
\renewcommand{\arraystretch}{1.2}
\begin{adjustbox}{max width=\columnwidth}
\begin{tabular}{lcccccccccc}
\toprule
\multirow{2}{*}{\textbf{Models}} 
  & \multicolumn{2}{c}{\textbf{Beauty}} 
  & \multicolumn{2}{c}{\textbf{Sports}} 
  & \multicolumn{2}{c}{\textbf{Toys}} 
  & \multicolumn{2}{c}{\textbf{ML-10M}}
  & \multicolumn{2}{c}{\textbf{ML-20M}} \\
\cmidrule(lr){2-3} \cmidrule(lr){4-5} \cmidrule(lr){6-7} \cmidrule(lr){8-9} \cmidrule(lr){10-11}
  & Train & Infer & Train & Infer & Train & Infer & Train & Infer & Train & Infer \\
\midrule
TiSASRec    & 5.09s  & 0.0001s & 7.18s  & 0.0001s & 4.45s  & 0.0001s  & 15.42s & 0.0001s & 35.15s & 0.0001s \\
BERT4Rec    & 5.30s  & 0.0001s & 8.48s  & 0.0001s & 3.73s  & 0.0001s  & 44.92s & 0.0001s & 81.99s & 0.0001s \\
SASRec      & 2.82s  & 0.0001s & 4.22s  & 0.0001s & 2.35s  & 0.0001s  & 11.19s & 0.0001s & 21.66s & 0.0001s \\
DiffuRec    & 3.56s  & 0.0044s & 5.66s  & 0.0045s & 3.10s  & 0.0045s  & 17.03s & 0.0057s & 35.15s & 0.0056s \\
DreamRec    & 12.73s & 0.0129s & 19.72s & 0.0118s & 10.74s & 0.0118s  & 153.50s & 0.0119s & 274.40s & 0.0211s \\
PreferDiff  & 13.66s & 0.0011s & 22.11s & 0.0011s & 12.15s & 0.0011s  & 139.29s & 0.0022s & 322.30s & 0.0014s \\
\modelname     & 15.13s & 0.0011s & 25.14s & 0.0012s & 13.18s & 0.0011s  & 147.32s & 0.0022s & 785.13s & 0.0054s \\
\bottomrule
\end{tabular}
\end{adjustbox}
\footnotesize{Note: \textit{Train} refers to the average time required to train one epoch, while \textit{Infer} indicates the average time for predicting the target IoI of a single sample.}
\end{table}

\subsection{Study of \modelname\ (RQ2 \& RQ3)}
\subsubsection{Time Encoding Function Analysis (RQ2)}
\label{subsubsec:Encoding_Analysis}
To answer RQ2, we construct a comparative experiment for three different time encoding functions and analyze their impact of performance. As shown in Table~\ref{tab:combined_encoding}, we observe that no single time encoding function consistently outperforms the others across all datasets. This suggests that temporal interaction patterns vary across datasets, and the effectiveness of time encoding functions should be considered dataset-dependent. Nevertheless, compared to absolute positional encoding, encoding real interaction timestamps consistently improves performance, regardless of the specific time encoding method. These findings highlight the importance of temporal information as a valuable signal for capturing the dynamics of users' personalized preferences and improving sequential recommendation accuracy.

\begin{table}[htbp]
\centering
\caption{Time Encoding Functions Performance Comparison}
\label{tab:combined_encoding}
\begin{adjustbox}{max width=\columnwidth}
\begin{tabular}{llcccccccc}
\toprule
\multirow{2}{*}{\textbf{Dataset}} & \multirow{2}{*}{\textbf{Functions}} & \multicolumn{4}{c}{\textbf{Leave-one-out Split}} & \multicolumn{4}{c}{\textbf{Temporal Split}} \\
\cmidrule(lr){3-6} \cmidrule(lr){7-10}
& & \textbf{H@5} & \textbf{H@10} & \textbf{N@5} & \textbf{N@10} & \textbf{H@5} & \textbf{H@10} & \textbf{N@5} & \textbf{N@10} \\
\midrule

\multirow{4}{*}{Beauty}
  & Position     & 3.00   & 3.81     & 2.23   & 2.49  & 4.07   & 4.74   & 3.10   & 3.32 \\
  & Sinusoidal   & 3.16   & 4.18     & 2.33   & 2.66  & \textbf{4.60}  & \textbf{5.63}  & \textbf{3.40}  & \textbf{3.74} \\
  & Gaussian     & 3.28   & 4.15     & 2.44   & 2.72  & 4.11   & 5.10  & 3.27   & 3.59 \\
  & RFF          & \textbf{3.36} & \textbf{4.15}  & \textbf{2.47} & \textbf{2.73} & 4.07   & 5.36  & 3.24   & 3.66 \\
\midrule

\multirow{4}{*}{Sports}
  & Position     & 1.50   & 1.93     & 1.06   & 1.20   & 1.85   & 2.19   & 1.47   & 1.58 \\
  & Sinusoidal   & 1.62   & 2.01     & 1.16   & 1.29   & \textbf{2.02}  & \textbf{2.30}  & \textbf{1.51}  & \textbf{1.60} \\
  & Gaussian     & 1.59   & 2.02     & 1.15   & 1.28   & 1.83 & 2.16 & 1.47 & 1.58 \\
  & RFF          & \textbf{1.66} & \textbf{2.14}  & \textbf{1.17} & \textbf{1.32} & 1.80 & 2.13 & 1.43 & 1.54 \\
\midrule

\multirow{4}{*}{Toys}
  & Position     & 3.21   & 3.66    & 2.49   & 2.63  & 4.27   & 5.05   & 3.31   & 3.56 \\
  & Sinusoidal   & 3.35   & 3.84    & 2.57   & 2.73  & \textbf{4.94}  & \textbf{5.82}  & \textbf{3.76}  & \textbf{4.05} \\
  & Gaussian     & \textbf{3.38} & \textbf{3.81}  & \textbf{2.68} & \textbf{2.82}  & 4.69  & 5.46 & 3.60  & 3.84 \\
  & RFF          & 3.36   & 3.81    & 2.67   & 2.82  & 4.69  & 5.56 & 3.56  & 3.84 \\
\midrule

\multirow{4}{*}{ML-10M}
  & Position     & 6.48   & 8.98    & 4.48  & 5.28   & 8.67   & 11.55   & 6.32   & 7.25 \\
  & Sinusoidal   & \textbf{8.24} & \textbf{10.74} & \textbf{6.01} & \textbf{6.82} & \textbf{10.80} & \textbf{13.38} & \textbf{8.16} & \textbf{8.99} \\
  & Gaussian     & 7.98   & 10.81   & 5.60  & 6.52   & 10.03  & 13.08   & 7.48   & 8.46 \\
  & RFF          & 8.17   & 11.09   & 5.70  & 6.64   & 9.90   & 12.82   & 7.31   & 8.27\\
\midrule

\multirow{4}{*}{ML-20M}
  & Position     & 6.36   & 8.87    & 4.35   & 5.16  & 8.19   & 10.92  & 5.98   & 6.87 \\
  & Sinusoidal   & \textbf{8.36} & \textbf{11.39} & \textbf{5.89} & \textbf{6.87} & \textbf{9.76} & \textbf{13.04} & \textbf{7.12} & \textbf{8.18} \\
  & Gaussian     & 8.15   & 11.19  & 5.67   & 6.65 & 9.57   & 12.71  & 7.03  & 8.04 \\
  & RFF          & 7.99   & 11.06  & 5.50   & 6.49 & 9.23   & 12.25  & 6.58  & 7.56 \\
\bottomrule
\end{tabular}
\end{adjustbox}
\footnotesize{Note: \textit{Position} uses absolute positions; \textit{Sinusoidal}, \textit{Gaussian}, and \textit{RFF} encode timestamps. \textbf{Bold values mark the best performance in each dataset.}}
\end{table}

\subsubsection{Ablation Study (RQ3)}
We construct an ablation study to evaluate the effectiveness of each component in \modelname. For fairness, we select the best-performing time encoding function for each dataset based on Table~\ref{tab:combined_encoding}. As shown in Table~\ref{tab:combined_ablation}, replacing absolute positional encoding with time encoding consistently improves performance, as also discussed in Section~\ref{subsubsec:Encoding_Analysis}. Building upon this, incorporating the ToI Prediction Module further improves performance across all datasets and evaluation metrics. These ablation results highlight the importance of modeling both past and future ToI information. More importantly, they demonstrate that the proposed components not only are individually beneficial, but also jointly contribute to performance improvements through integration within \modelname.

\begin{table}[htbp]
\centering
\caption{Ablation Study of \modelname\ Components}
\label{tab:combined_ablation}
\begin{adjustbox}{max width=\columnwidth}
\begin{tabular}{llcccccccc}
\toprule
\multirow{2}{*}{\textbf{Dataset}} & \multirow{2}{*}{\textbf{Setting}} & \multicolumn{4}{c}{\textbf{Leave-one-out Split}} & \multicolumn{4}{c}{\textbf{Temporal Split}} \\
\cmidrule(lr){3-6} \cmidrule(lr){7-10}
& & \textbf{H@5} & \textbf{H@10} & \textbf{N@5} & \textbf{N@10} & \textbf{H@5} & \textbf{H@10} & \textbf{N@5} & \textbf{N@10} \\
\midrule
\multirow{3}{*}{Beauty}
 & Base         & 3.00   & 3.81      & 2.23   & 2.49   & 4.07   & 4.74   & 3.10   & 3.32 \\
 & +TE          & 3.36↑  & 4.15↑     & 2.47↑  & 2.73↑  & 4.60↑  & 5.63↑  & 3.40↑  & 3.74↑ \\
 & +TE+TP       & 3.43↑  & 4.31↑     & 2.48↑  & 2.77↑  & 4.74↑  & 5.63↑  & 3.58↑  & 3.87↑ \\
\midrule
\multirow{3}{*}{Sports}
 & Base         & 1.50   & 1.93   & 1.06   & 1.20   & 1.85   & 2.19   & 1.47   & 1.58 \\
 & +TE          & 1.66↑  & 2.14↑  & 1.17↑  & 1.32↑  & 2.02↑  & 2.30↑  & 1.51↑  & 1.60↑ \\ 
 & +TE+TP       & 1.66↑  & 2.14↑  & 1.21↑  & 1.36↑  & 2.11↑  & 2.50↑  & 1.65↑  & 1.78↑ \\
\midrule
\multirow{3}{*}{Toys}
 & Base         & 3.21   & 3.66   & 2.49   & 2.63   & 4.27   & 5.05   & 3.31   & 3.56 \\
 & +TE          & 3.38↑  & 3.81↑  & 2.68↑  & 2.82↑  & 4.94↑  & 5.82↑  & 3.76↑  & 4.05↑ \\
 & +TE+TP       & 3.39↑  & 3.88↑  & 2.65↑  & 2.81↑  & 4.94↑  & 5.61↑  & 3.84↑  & 4.06↑ \\
\midrule
\multirow{3}{*}{ML-10M}
 & Base         & 6.48   & 8.98    & 4.48   & 5.28   & 8.67   & 11.55   & 6.32   & 7.25 \\
 & +TE          & 8.24↑  & 10.74↑  & 6.01↑  & 6.82↑  & 10.80↑ & 13.38↑  & 8.16↑  & 8.99↑ \\
 & +TE+TP       & 8.89↑  & 11.67↑  & 6.43↑  & 7.33↑  & 11.00↑ & 13.79↑  & 8.31↑  & 9.21↑ \\
\midrule
\multirow{3}{*}{ML-20M}
 & Base         & 6.36   & 8.87    & 4.35   & 5.16   & 8.19    & 10.92   & 5.98   & 6.87 \\
 & +TE          & 8.36↑  & 11.39↑  & 5.89↑  & 6.87↑  & 9.76↑   & 13.04↑  & 7.12↑  & 8.18↑ \\
 & +TE+TP       & 8.96↑  & 11.95↑  & 6.39↑  & 7.35↑  & 10.56↑  & 13.41↑  & 7.81↑  & 8.74↑ \\
\bottomrule
\end{tabular}
\end{adjustbox}
\footnotesize{Note: \textit{Base} refers to the model using absolute positional encoding. \textit{+TE} indicates the replacement of absolute positions with time encoding function. \textit{+TP} denotes the incorporation of the ToI Prediction Module. ↑ indicates improvement.}
\end{table}

\subsection{Impact of ToI information (RQ4)}
To answer RQ4, we analyze the impact of incorporating ToI information into user representations using the Amazon Beauty and Toys datasets, with H@5 as the evaluation metric. As shown in Figure~\ref{fig:Parameters}, the model achieves optimal performance when the residual weight \( \gamma \) is relatively large (e.g., 0.7 or 0.8), with slight variations influenced by the loss weight \( \eta \). These results suggest that incorporating a higher proportion of predicted ToI information into user representations provides enriched guidance during inference. More ToI signal generally provides stronger time guidance for inference, but the specific value should be tuned regarding the real-data distribution.

\begin{figure}[htbp]
    \centering
    \includegraphics[width=\columnwidth]{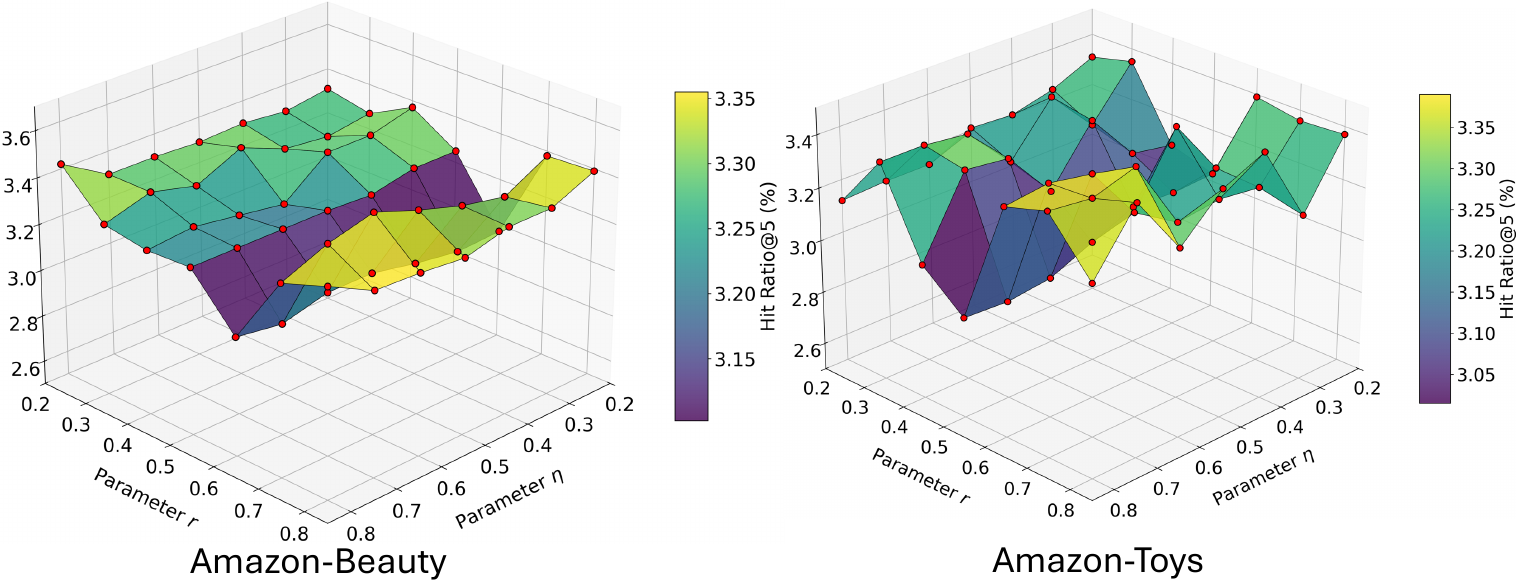}  
    \caption{The Impact of Incorporating ToI Information}
    \label{fig:Parameters}                    
\end{figure}

\subsection{Accuracy Analysis of ToI Prediction (RQ5)}
\label{ToIaccuracy}
To evaluate the accuracy of ToI prediction, we measure the cosine similarity between the predicted and ground-truth ToI embeddings on the Amazon Beauty and MovieLens 20M test sets under the LOO partition. As shown in Figure~\ref{fig:TimeSimilarity}, we observe that the well-trained ToI prediction module performs remarkably well, accurately predicting the next ToI embedding for most samples based on their historical interaction sequences. In particular, on both datasets, the cosine similarity exceeds 0.95 for the majority of test samples, demonstrating the precise ToI prediction capability of \modelname.

\begin{figure}[htbp]
    \centering
    \includegraphics[width=\columnwidth]{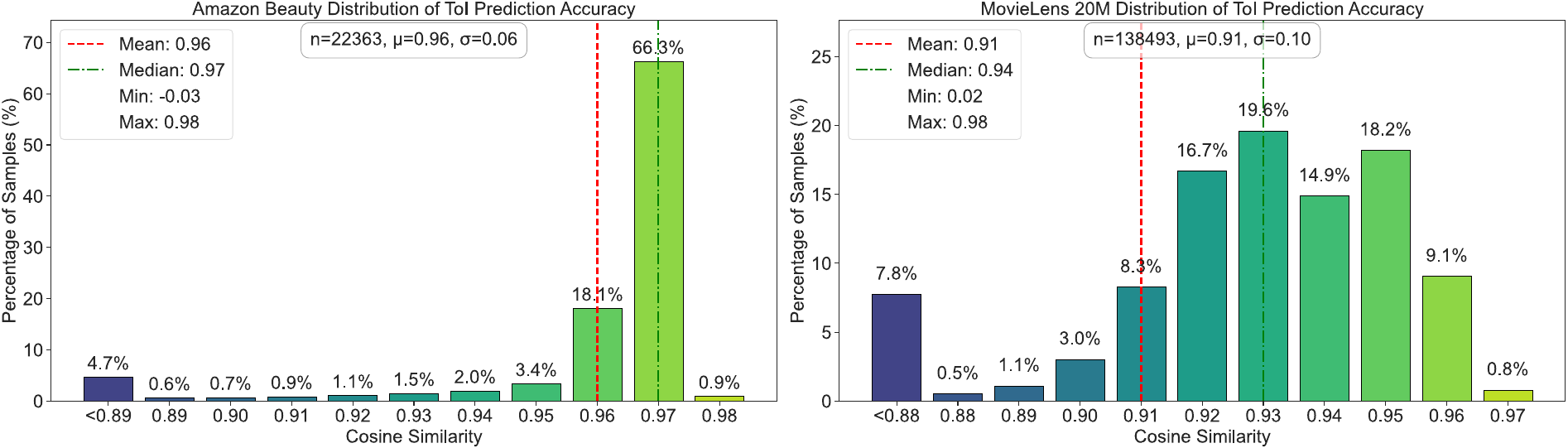}  
    \caption{The Accuracy Distribution of ToI Prediction}
    \label{fig:TimeSimilarity}                    
\end{figure}

\section{Conclusion}

In this paper, we investigate the importance of jointly modeling users’ Time of Interest (ToI) and Item of Interest (IoI) to enhance sequential recommendation, particularly in the context of active recommendation, where the system must not only predict what to recommend but also when to recommend. This distinct our work from most existing sequential recommenders that solely focus on what to recommend.
Specifically, we propose \modelname, a diffusion-based framework supported by theoretical analysis that shows integrating predicted ToI into user representations leads to a tighter evidence lower bound. With carefully designed modules and training objectives, \modelname\ effectively captures the interaction between timing and item preferences to deliver accurate and timely recommendations. 
Extensive experiments on five benchmark datasets under both leave-one-out and temporal data splits demonstrate that \modelname\ consistently outperforms eight state-of-the-art baselines. These results confirm the effectiveness of learning when and what to recommend—paving the way for more personalized active recommender systems.

\clearpage
\bibliographystyle{plain}
\bibliography{references}

\clearpage

\appendix
\section{Experiment Implementation Details}
\label{app:exp}
All experiments were conducted in PyTorch on a single NVIDIA Tesla V100-SXM2 GPU with 32GB of memory. We optimized all models using the AdamW optimizer with an initial learning rate of 0.0003. All parameters, including item embeddings, were randomly initialized.  
For data partitioning, we adopted two strategies. In the 8:1:1 temporal split, each dataset was randomly divided into training, validation, and test sets in an 8:1:1 ratio. In the LOO setting, the most recent interaction of each user was used for testing, the second most recent for validation, and the remaining interactions for training.  
We set the training batch size to 256 and the evaluation batch size to 32. Model validation was performed after each epoch, and early stopping was applied if no improvement was observed for 10 consecutive epochs.
In our experiments, we strictly follow the official implementations of all baseline methods, and their corresponding code repositories are listed as follows:
\begin{itemize}[leftmargin=*,itemsep=0.5pt,topsep=2.5pt]
    \item GRU4Rec: \url{https://github.com/RUCAIBox/RecBole/blob/master/recbole/model/sequential_recommender/gru4rec.py}
    \item SASRec: \url{https://github.com/pmixer/SASRec.pytorch}
    \item BERT4Rec: \url{https://github.com/RUCAIBox/RecBole/blob/master/recbole/model/sequential_recommender/bert4rec.py}
    \item TiSASRec: \url{https://github.com/pmixer/TiSASRec.pytorch}
    \item MEANTIME: \url{https://github.com/SungMinCho/MEANTIME}
    \item DreamRec: \url{https://github.com/YangZhengyi98/DreamRec}
    \item DiffuRec: \url{https://github.com/WHUIR/DiffuRec}
    \item PreferDiff: \url{https://github.com/lswhim/PreferDiff}
\end{itemize}

\section{Hyper-parameters Settings}
\label{app:hp}
To ensure fair and consistent evaluation, we conducted extensive hyperparameter tuning based on empirical performance. The coefficient \( \lambda \) controls the trade-off between learning to generate item embeddings and capturing user preference, and is used in both PreferDiff and \modelname. The parameter \( w \) determines the strength of the personalized guidance signal during the inference stage, and is applied in DreamRec, PreferDiff, and \modelname. Following the previous work~\cite{liu2024preference}, we search \( \lambda \) from the set \( \{0.1, \dots, 0.9\} \), and \( w \) from \( \{0,\dots,9\} \). The best-performing configurations for each dataset are as follows: Beauty uses \( \lambda = 0.4 \), \( w = 8 \); Sports uses \( \lambda = 0.4 \), \( w = 2 \); Toys uses \( \lambda = 0.6 \), \( w = 6 \); ML-10M uses \( \lambda = 0.8 \), \( w = 4 \); and ML-20M uses \( \lambda = 0.7 \), \( w = 5 \).

For time encoding methods based on Gaussian and RFF, we additionally tune the frequency scale parameter \( \sigma \), selected from the candidate set \( \{1, 0.5, 0.1, 0.05, 0.01, 0.001\} \). We find that Gaussian embedding achieves the best performance with \( \sigma = 0.05 \), while RFF embedding performs best with \( \sigma = 1.0 \).

We search \( \gamma \) and \( \eta \) from the set \( \{0.1, 0.2, \dots, 1.0\} \). The coefficient \( \gamma \) balances the contribution of the predicted ToI information to the user representation using a residual connection. The coefficient \( \eta \in [0,1] \) balances the model's ability to denoise and reconstruct the target IoI embedding and to predict the potential ToI embedding. The best-performing configurations for each dataset under the LOO partition are as follows: Beauty uses \( \gamma = 0.8 \), \( \eta = 0.2 \); Sports uses \( \gamma = 0.4 \), \( \eta = 0.8 \); Toys uses \( \gamma = 0.3 \), \( \eta = 0.5 \); ML-10M uses \( \gamma = 0.6 \), \( \eta = 0.2 \); and ML-20M uses \( \gamma = 0.3 \), \( \eta = 0.5 \). And under 8:1:1 temporal splitting partition are as follows: Beauty uses \( \gamma = 0.3 \), \( \eta = 0.1 \); Sports uses \( \gamma = 1.0 \), \( \eta = 0.8 \); Toys uses \( \gamma = 0.1 \), \( \eta = 0.9 \); ML-10M uses \( \gamma = 0.7 \), \( \eta = 0.2 \); and ML-20M uses \( \gamma = 0.4 \), \( \eta = 0.9 \).

\section{Algorithm}
\label{app:alg}
We present the algorithms for the training and inference processes of \modelname\ as follows:
\begin{algorithm}[hbp]
\caption{Training Algorithm of \modelname}
\label{alg:tagdiff_train}
\begin{algorithmic}[1]
\Require Historical interaction sequence $s_u$, interacted time $t_u$, target IoI $e_n$, Time Encoding function $TE$, Transformer Encoder $TR$, hyperparameters $\gamma,\eta$, learning rate $lr$, variance schedule $\{\alpha_t\}_{t=1}^T$
\Ensure Optimal denoising model $p_\theta(\cdot)$ and optimal ToI prediction module $M_\varphi(\cdot)$
\Repeat
    \State $t \sim \{1, \dots, T\}$, $\epsilon \sim \mathcal{N}(0, I)$ 
    \State $e_n^t = \sqrt{\bar{\alpha}_t} e_n + \sqrt{1 - \bar{\alpha}_t} \, \epsilon$ \Comment{Add Gaussian Noise}
    \State $\tilde{s}_u = s_u + TE(t_u)$ \Comment{Encode timestamps into sequence}
    \State $g_u = TR(\tilde{s}_u)$ \Comment{Extract user representation}
    \State $\hat\tau_{n} = M_\varphi(g_u, \tau_{n-1})$ \Comment{Predict next ToI embedding}
    \State $g_u' = (1 - \gamma)\,g_u + \gamma \cdot \mathrm{MLP}(g_u, \hat\tau_{n})$ 
    \State $\mathcal{L}_{\text{ToI}}, \mathcal{L}_{\text{IoI}} \gets$ compute via Equations~\eqref{equation_Ltime} and~\eqref{equation_LDiff}
    \State $\mathcal{L}_{\text{\modelname}} = \eta \mathcal{L}_{\text{IoI}} + (1 - \eta)\mathcal{L}_{\text{ToI}}$
    \State $\theta \gets \theta - lr \nabla_\theta \mathcal{L}_{\text{\modelname}}$, \quad 
       $\varphi \gets \varphi - lr \nabla_\varphi \mathcal{L}_{\text{\modelname}}$
\Until{converged}
\end{algorithmic}
\end{algorithm}

\begin{algorithm}[hbp]
\caption{Inference Algorithm of \modelname}
\label{alg:tagdiff_infer}
\begin{algorithmic}[1]
\Require Historical interaction sequence $s_u$, interacted time $t_u$, Time Encoding function $TE$, Transformer Encoder $TR$, hyperparameters $\{\gamma, w\}$, variance schedule $\{\alpha_t\}_{t=1}^T$, optimal denoising model $p_\theta(\cdot)$, optimal ToI prediction module $M_\varphi(\cdot)$, dummy token $\Phi$
\Ensure Predicted target IoI $\hat{e}_n^{0}$
\State $e_n^T \sim \mathcal{N}(0, I)$ \Comment{Initialize Gaussian noise}
\State $\tilde{s}_u = s_u + TE(t_u)$ \Comment{Encode timestamps into sequence}
\State $g_u = TR(\tilde{s}_u)$ \Comment{Extract user representation}
\State $\hat\tau_{n} = M_\varphi(g_u, \tau_{n-1})$ \Comment{Predict next ToI embedding}
\State $g_u' = (1 - \gamma)\,g_u + \gamma \cdot \mathrm{MLP}(g_u, \hat\tau_{n})$ \Comment{Enrich user representation}
\For{$t = T$ \textbf{down to} $1$}
    \State $\hat{e}_n^{0} = (1 + w) \cdot p_\theta(e_n^{t}, t, g_u') - w \cdot p_\theta(e_n^{t}, t, \Phi)$
    \State $e_n^{t-1} = \sqrt{\bar{\alpha}_{t-1}} \, \hat{e}_n^{0} 
        + \sqrt{1 - \bar{\alpha}_{t-1}} \, 
        \frac{e_n^{t} - \sqrt{\bar{\alpha}_t} \, \hat{e}_n^{0}}{\sqrt{1 - \bar{\alpha}_t}}$
\EndFor
\State \Return $\hat{e}_n^{0}$
\end{algorithmic}
\end{algorithm}

\end{document}